\documentclass{article}

\usepackage{primearxiv}

\usepackage[utf8]{inputenc} 
\usepackage[T1]{fontenc}    
\usepackage{url}            
\usepackage{booktabs}       
\usepackage{amsfonts}       
\usepackage{nicefrac}       
\usepackage{microtype}      
\usepackage{lipsum}

\usepackage{graphicx}       
\graphicspath{{media/}}     

\usepackage{todonotes}
\definecolor{emilioeditcolor}{rgb}{0.94, 0.97, 1.0}

\usepackage{theoremref}

\usepackage{amsfonts, amsthm, amsmath, amssymb,inputenc,geometry,todonotes,mathtools,physics, bbm, tikz-cd}
\usepackage[toc,page]{appendix}
\usepackage{array}
\usepackage{kpfonts}
\usepackage{yfonts}
\usepackage[mathscr]{euscript}
\definecolor{winered}{rgb}{0.5,0,0}
\usepackage{hyperref}

\newtheoremstyle{theoremdd}
{\topsep}{\topsep}{\upshape}{0pt}{\bfseries}{.}{ }{\thmname{#1}\thmnumber{ #2}\thmnote{ (#3)}}
  
\theoremstyle{theoremdd}
\newtheorem{Th}{Theorem}[section]

\theoremstyle{definition}
\newtheorem{Lemma}[Th]{Lemma}
\newtheorem{Cor}[Th]{Corollary}
\newtheorem{Prop}[Th]{Proposition}
\newtheorem{Def}[Th]{Definition}

\newtheorem{Ex}[Th]{Example}

\newcommand{\R}{\mathbb{R}}
\newcommand{\Z}{\mathbb{Z}}

\usetikzlibrary{calc}
\usetikzlibrary{decorations.pathmorphing}

\tikzset{curve/.style={settings={#1},to path={(\tikztostart)
    .. controls ($(\tikztostart)!\pv{pos}!(\tikztotarget)!\pv{height}!270:(\tikztotarget)$)
    and ($(\tikztostart)!1-\pv{pos}!(\tikztotarget)!\pv{height}!270:(\tikztotarget)$)
    .. (\tikztotarget)\tikztonodes}},
    settings/.code={\tikzset{quiver/.cd,#1}
        \def\pv##1{\pgfkeysvalueof{/tikz/quiver/##1}}},
    quiver/.cd,pos/.initial=0.35,height/.initial=0}

\tikzset{tail reversed/.code={\pgfsetarrowsstart{tikzcd to}}}
\tikzset{2tail/.code={\pgfsetarrowsstart{Implies[reversed]}}}
\tikzset{2tail reversed/.code={\pgfsetarrowsstart{Implies}}}

\usepackage{tkz-graph}
\tikzstyle{vertex}=[circle, draw, inner sep=0pt, minimum size=6pt]

\title{Regularizing Calabi-Yau topological conformal field theories using cutoff heat kernels}

\author{
  Yashasvi Aulak \\
  Graduate Center, City University of New York
 }

\begin{document}
\maketitle

\begin{abstract}
In this paper we construct a family of topological conformal field theories (TCFTs) associated to a Calabi-Yau space by modifying the heat kernel and sections of the Calabi-Yau space. This is done by restricting to certain eigenspaces of the Laplacian. We then present two a-priori distinct ways to regularize the Calabi-Yau TCFT by using these modified heat kernels, and then show that they are equivalent. Finally, we relate the regularized TCFTs for different cutoffs.
\end{abstract}

\section{Introduction}

Topological conformal field theories (TCFTs) have played an important role in the world of mathematical physics for the past few decades. In physics, TCFTs provide models for condensed matter physics while in mathematics they are used to calculate certain topological invariants. A key component of quantum field theory (QFT) in general is renormalization, which renders the theory useful for an experimental physicist by obtaining finite results from divergent integrals. Initially, renormalization seemed like an ad-hoc fix that lacked a solid mathematical foundation. However, a few decades later, the reinterpretation of QFT as a scale dependent theory established renormalization as an intrinsic component of QFT. Kevin Costello in his book \cite{book}, gives a beautiful and rigorous account of renormalization. In this work, we apply renormalization techniques to regularize TCFTs associated to a Calabi-Yau space. 

We start with a concise background of TCFTs as described by Kevin Costello in \cite{tcft}. The Segal category of metrized ribbon graphs is a differential graded symmetric monoidal (dgsm) category whose objects are non-negative integers and morphisms are metrized ribbon graphs with a number of incoming and outgoing labeled external edges. A TCFT is a dgsm functor from the Segal category of metrized ribbon graphs to the dgsm category of chain complexes. The data of a Calabi-Yau space consists of an oriented Riemannian manifold $M$ equipped with a bundle of $\mathbb{Z}/2$-graded complex algebras $\mathcal{A}_M$, whose algebra of global sections we call $\mathcal{A}$, along with a trace map compatible with the Riemannian metric on $M$. This data can be used to construct a TCFT as in \cite{tcft}. 

Let us recall this construction. First, let $\Gamma(m,n)$ denote the moduli space of metrized ribbon graphs with $m$ incoming and $n$ outgoing external vertices. A singular cubical chain in $\Gamma(m,n)$ determines an unmetrized ribbon graph $\gamma$. Similar to Feynman graphs, the edges of $\gamma$ can be labeled with propagator-like forms constructed using the heat kernel associated to the Laplacian on $M$. The vertices of $\gamma$ are labeled with functionals that integrate volume forms on $M$. The incoming external vertices are labeled with input forms in $\mathcal{A}^{\otimes m}$ and after multiplying and integrating forms using Feynman diagrams-like rules, one gets output forms in $\mathcal{A}^{\otimes n}$. This gives a map from singular cubical chains on the moduli space $C_\bullet(\Gamma(m,n))$ to $\text{Hom}(\mathcal{A}^{\otimes m}, \mathcal{A}^{\otimes n})$ which determines the TCFT. 

The integrals involved in the process described above are not always convergent. Costello imposes some additional restrictions on the moduli space $\Gamma(m,n)$. One of these restrictions is that there are no graphs with loops of vanishing length, which ensures that the integrals always converge. However, such a moduli space is not compact. Our main goal (not achieved in this work), is to construct a well defined TCFT on a compactified moduli space. In order to compactify, we allow graphs with shrinking loops and then use renormalization in QFT to get finite results.   

In this work we present a way to impose regulators on the TCFT. We do this in two a-priori distinct ways, and ultimately prove that they are equivalent. We then establish a simple relationship among the regularized TCFTs with different cutoffs. 

In \cite{tcft}, Costello describes a length based regulator dependent TCFT and draws connections to QFT. In future work, we will study that further and construct a regulator independent TCFT.

The paper is organized as follows:
\begin{itemize}
    \item in section \ref{section background} of this paper we give the relevant background about Calabi-Yau spaces, the moduli space of metrized ribbon graphs, dgsm categories and TCFTs,
    \item in section \ref{Single eigenvalue TCFTs} we introduce a way to modify the heat kernel by restricting to a single eigenspace of the Laplacian. We then show that using the modified heat kernel and Calabi-Yau space, we get a TCFT, which we call a single-eigenvalue TCFT,
    \item in section \ref{Constructing TCFTs using single-eigenvalue TCFTs} we introduce two ways to construct TCFTs using single-eigenvalue TCFTs. The first way is by defining a sum of single-eigenvalue TCFTs. The second way is by restricting to a subspace built using multiple eigenspaces. We then prove that these constructions are equivalent. Finally we relate these regularized TCFTs to QFTs regularized with sharp cutoffs. Appendix \ref{section category theory background} provides the relevant category theory background.
    
\end{itemize}

\section*{Acknowledgments}

I am thankful to my PhD advisor Mahmoud Zeinalian for the numerous conversations about this project.\\
I would also like to thank Jeffrey Kroll and Emilio Minichiello for very helpful discussions on differential geometry and category theory.

\section{Background} \label{section background}

Sections \ref{32}-\ref{36} are adapted from \cite{tcft}. We mostly follow the notation from that reference as well. 

\subsection{Calabi-Yau Spaces} \label{32}

In this section we recall the definition of a Calabi-Yau space. First, we need a few preliminary definitions.

\begin{Def}
An \textit{elliptic space} $(M,\mathcal{A})$ consists of
\begin{enumerate}
        \item an $m$-dimensional compact (oriented) Riemannian manifold M, with $m \geq 2$,
        \item a bundle $\mathcal{A}_M$ of complex $\Z_2$-graded associative algebras over $C^\infty(M)$, whose algebra of global sections we denote by $\mathcal{A}$, and
        \item an odd derivation $Q : \mathcal{A} \to \mathcal{A}$ such that
        \begin{itemize}
            \item $Q^2=0$,
            \item $Q$ is a degree 1 differential operator on $C^\infty(M)$, and
            \item $(\mathcal{A},Q)$ is an elliptic complex. 
        \end{itemize}
    \end{enumerate}
A \textit{Calabi-Yau structure} on an elliptic space $(M, \mathcal{A})$ consists of a \textit{trace} map $\text{Tr}: \mathcal{A} \to \mathbb{C}$ which factors as
\begin{equation*}
    \mathcal{A} \xrightarrow{\text{Tr}_M} \Omega^m(M) \xrightarrow{\int_M} \mathbb{C},
\end{equation*}
and for each $e, e' \in \mathcal{A}$ satisfies
\begin{equation*}
    \text{Tr}[e,e']=0, \qquad \text{ and } \qquad  \text{Tr}(Qe)=0.
\end{equation*}
This gives a trace pairing $\text{Tr}:\mathcal{A} \otimes \mathcal{A} \to \Omega^m(M) $, and we require it to be fiber-wise non-degenerate. (Wherever appropriate, $\mathcal{A}_M \otimes \mathcal{A}_M$ and $\mathcal{A} \otimes \mathcal{A} $ denote the bundle $\mathcal{A}_M \boxtimes \mathcal{A}_M \to M \cross M$ and sections of the bundle respectively, referring to the completed projective tensor product). 
    
A Hermitian metric on $\mathcal{A}$ is called compatible if there exists a linear operator $*:\mathcal{A} \to \mathcal{A}$ such that $\text{Tr}(e*e')=<e,e'>$ and $**e=(-1)^P e$. Here, the parity P is 1 or 0 depending on whether the trace map is odd or even.
\end{Def}

The following proposition is stated without proof (see \cite{tcft} for proof).
\begin{Prop}
Let $Q^\dagger$ be the Hermitian adjoint to $Q$ with respect to the metric. Define the Laplacian/Hamiltonian by $H=[Q,Q^\dagger].$ \\ 
Then, 
\begin{itemize}
    \item $H$ is elliptic, self adjoint with respect to $<,>$ and has non-negative eigenvalues. The eigenvectors of H span $L^2$ completion of $\mathcal{A}$. 
    \item $Q^\dagger$, $H$ are both self adjoint with respect to Trace pairing.
    \item $\mathcal{A} = \text{Im}(Q) \oplus \text{Ker}(H) \oplus \text{Im}(Q^\dagger)$ and $\text{Ker}(H) = \text{Ker}(Q) \cap \text{Ker}(Q^\dagger)$ 
\end{itemize}
\end{Prop}
\begin{Ex}
    An example of a Calabi-Yau space and the setup that we work in for the rest of this work is;\\
    (M is a ($d \geq 2$) dimensional compact (oriented) Riemannian manifold, $\mathcal{A}=\Omega^\bullet(M,\R),Q=d$, $\text{Tr}_x(a(x)b(x)) = \int_{x \in M} a(x)\wedge b(x)$, $ \langle , \rangle $ is the Euclidean metric induced on differential forms, $Q^\dagger=d^*, * = \star$ is the Hodge star operator,  and $H := dd^* + d^* d)$. Henceforth, this tuple will be referred to as the \textit{data of a Calabi-Yau space}.
\end{Ex}

\subsection{The heat kernel}   \label{33} 

Given a general Calabi-Yau space, one can study the heat kernel for the Laplacian $H$ (as in \cite{tcft}),
\begin{Def}
\label{31}
    A \textit{heat kernel for $H$} is an element,
    \begin{equation*}
        K \in \Gamma(\R_{\geq 0} \cross M \cross M, \mathcal{A}_m \boxtimes \mathcal{A}_m)
    \end{equation*}
    such that,
    \begin{equation}
        \frac{d}{dt}K_t(x,y)=-H_xK_t(x,y)
    \end{equation}
    \begin{equation}
        \text{lim}_{t \to 0} K_t(x,y) = \delta_{x,y}
    \end{equation}
\end{Def}
where $\delta_{x,y}$ is the delta distribution. Define $L_t(x,y) = -Q_x^\dagger K_t(x,y)$. The following proposition is stated without proof (the proof can be found in \cite{tcft})
\begin{Prop}
\label{4}
    $K_t(x,y)$ and $L_t(x,y)$ satisfy the following relations,
    \begin{equation}
        [Q,C_{K_t(x,y)}]=[Q^\dagger,C_{K_t(x,y)}]=[H,C_{K_t(x,y)}]=0 \hspace{5pt} ; \hspace{5pt} [Q,C_{L_t(x,y)}] = -H K_t
    \end{equation}
    where $C_K : \mathcal{A} \to \mathcal{A}$ ; $ K \in \mathcal{A}^2$ is a convolution operator defined as $C_{K(x,y)}f:=(-1)^P \text{Tr}_yK(x,y)f(y)$ for $f \in \mathcal{A}$, and $P$ is the degree of trace map.
    \begin{equation}
        Q_xK_t(x,y) + Q_yK_t(x,y) = 0 \hspace{10 pt}  \text{ $K(x,y)$ is a closed differential form on $M \cross M$ $\forall t \in \R_{\geq 0}$} 
    \end{equation}
\begin{equation}
        Q^\dagger_xK_t(x,y) = Q^\dagger_yK_t(x,y) 
    \end{equation}
    \begin{equation}
       H_xK_t(x,y) = H_yK_t(x,y) 
    \end{equation}
    \begin{equation}
      K_t(x,y) = (-1)^PK_t(y,x) 
    \end{equation}
    \begin{equation}
      L_t(x,y) = (-1)^PL_t(y,x) 
    \end{equation}
     \begin{equation}
       (Q_x+Q_y) K_t(x,y)=-H_xK_t(x,y) = \frac{d}{dt}K_t(x,y) 
    \end{equation}
    
\end{Prop}

\subsection{Moduli space of metrized ribbon graphs} \label{34}
A TCFT is based on the Segal category of metrized ribbon graphs, whose set of morphisms is the moduli space of metrized ribbon graphs. We describe the moduli space in this subsection.
\begin{Def}
    A \textit{ribbon graph} (by Hamilton \cite{ham}) is a set $\Gamma$ (called the set of half-edges) together with the following data:
    \begin{itemize}
        \item A partition of $\Gamma$ into pairs, denoted by $E(\Gamma)$, called the set of edges of $\Gamma$.
        \item A partition of $\Gamma$, denoted by $V(\Gamma)$, called the set of vertices of $\Gamma$. We will refer to the cardinality of a vertex $v \in V(\Gamma)$ as the valency of $v$.
        \item A cyclic ordering on elements of every vertex $v \in V(\Gamma)$.
        \end{itemize}
A ribbon graph with a non-negative real number $l(e) \in \R_{\geq 0}$ (called length) assigned to each edge $e$ is called a \textit{metrized ribbon graph}. 
       \end{Def}

\begin{Def}
\cite{tcft} 
The \textit{moduli space of metrized ribbon graphs}  \footnote{This space is equivalent to the moduli space of Riemann surfaces, check \cite{tcft}].} with n incoming and m outgoing external vertices (denoted by $\Gamma(n,m)$) is a topological space defined (constructed) as follows,
\begin{enumerate}
\item  $\Gamma(n,m)$ as a set is the collection of all metrized ribbon graphs with n incoming and m outgoing labeled external vertices. A vertex is called external if it's one-valent. Non-external vertices are at least trivalent and are called internal vertices. External vertices are labeled either outgoing or incoming and are ordered (by definition). An edge connected to any external vertex is called an external edge (other edges are called internal). These graphs need not be connected.
\item $\Gamma(n,m)$ has the following non-compact orbi-cell decomposition. 
\begin{itemize}
\item Given an umetrized ribbon graph $\gamma$ with p edges, m,n incoming and outgoing vertices, a p-cell $O_\gamma$ labeled by $\gamma$ is a subset of $\Gamma(n,m)$ consisting of all metrized ribbon graphs which, upon ignoring
the length assignments, are the same as  $\gamma$.
\item The cell $O_\gamma$ can be identified with a p-dimensional Euclidean cone $\R^p_{\geq 0}$. Each point in $\R^p_{\geq 0}$ is an ordered tuple of p non-negative numbers, that along with $\gamma$ gives a unique metrized ribbon graph in $O_\gamma$. Give $O_\gamma$ the standard topology of $\R^p_{\geq 0}$.
\item Let $e$ be an internal edge in an metrized ribbon graph $\gamma$. If $l(e)=0$, then identify $\gamma$ with $\gamma/e$. Hence, for an unmetrized graph $\gamma$ the boundary of the cell $O_\gamma$ can be defined as the union of all the subsets $O_{\gamma/e}$ of $\Gamma(m,n)$, for each $e$. 
\item One can now form the p-skeleton of $\Gamma(m,n)$ by gluing all p cells $O_\gamma$ together, in a fashion that whenever any 2 (or more) p cells have an identical boundary component (a p-1 cell), we glue the cells along their boundaries. p+1 (and further) skeletons are formed this way and () Topology is given to the whole complex $\Gamma(m,n)$. 
\end{itemize}
\item The following two conditions are imposed on these metrized graphs.
\begin{itemize}
\item  Every closed loop of edges in $\gamma$ is of positive length.
\item Every path of edges in $\gamma$ which starts and ends at different outgoing external
vertices is of positive length.
\end{itemize}
\item The connected components of graphs are allowed to be of the following exceptional types.
\begin{itemize}
\item A graph with two external vertices, one edge, and no internal vertices. The external vertices can be any configuration of incoming or outgoing.
\item A graph with two internal vertices, and two edges between them. 
\end{itemize}
\end{enumerate}
\end{Def}

\begin{Def} 
   The space $\Omega^i(\Gamma(n,m))$ of cellular differential forms with values in $\R$ is defined as follows (\cite{tcft}). An element $\omega \in \Omega^i(\Gamma(n,m))$ consists of a form $\omega_\gamma \in \Omega^i(\text{Met} (\gamma)) :=  \Omega^i(\R^p_{\geq 0})$, for each graph $\gamma$ as above, which is Aut($\gamma$) equivariant, and such that, for each edge e of $\gamma$ which is not a loop, $\omega_{\gamma/e} = \omega|_{\text{Met}(\gamma/e)} $. Hence, forms defined on boundaries of adjacent p-cells agree, giving a form on p-skeleton (and inductively on $\Gamma(m,n)$)\\
For a fixed edge $e$ of an unmetrized graph $\gamma$, $l_e$ is a positive real valued function on the cell  $O_\gamma$ mapping metrized ribbon graphs length of edge $e$ and $dl_e \in \Omega^1(\text{Met}(\gamma))$ is a one-form on the cell. 
\end{Def}
\begin{Def}
Let $C_*(\Gamma(n,m))$ denote the subcomplex of normalised singular cubical chains on $\Gamma(n,m)$) spanned by simplices which are smooth and lie entirely in one of the closed cells.\\
Explicitly, a pure k-chain $\sigma \in C_*(\Gamma(n,m))$ lying completely inside a cell $O_\gamma$ is a map $ \sigma : \R^k_{\geq 0} \to O_\gamma \subset \Gamma(m,n)$. \\
There is an integration pairing ($ C_*(\Gamma(n,m)) \otimes \Omega^*_{cell}(\Gamma(n,m)) \to \R$) between forms and chains on $\Gamma(n,m)$ given by integrating a form on a chain.
\end{Def}

\subsection{labeling graphs in moduli space with  Calabi-Yau space valued forms} \label{35}

an unmetrized ribbon graph is used to construct the functor on morphisms, by taking input forms and via a Feynman rules like procedure, producing output forms. We describe the labeling procedure \cite{tcft} explicitly in this subsection. 

Let $\gamma$ be an unmetrized ribbon graph with n incoming and m outgoing labeled external vertices. $\text{Met}(\gamma)$ is the space of metrics on $\gamma$. \\ 
Let $a=a_1(x_1) \otimes .. \otimes a_n(x_n)$ where each $a_i(x_i) \in \mathcal{A}$ labels an incoming vertex with respect to labeling of incoming external vertices. 

\begin{Def}
Given an edge $e \in E(\gamma)$ joining $v_i,v_j \in V(\gamma)$, we label the edge by the form $\omega_e$ which is defined in the following way, \\ 
If $e$ is an internal edge, then,
\begin{equation*}
    \omega_e := K_t(x_i,x_j) + dt L_t(x_i,x_j) \in \mathcal{A}^{\otimes 2} \otimes \Omega^\bullet(\text{Met}(\gamma))
\end{equation*}
If $e$ is an incoming external edge with external vertex labeled by $a_j(x_j)$, then
\begin{equation*}
    \omega_e := [K_t(x_1,x_2) + dt L_t(x_1,x_2)] \wedge a_k(x_k) \in \mathcal{A}^{\otimes 2} \otimes \Omega^\bullet(\text{Met}(\gamma))
\end{equation*}
\end{Def}
\begin{Def}
    Define the map $ \tilde{K}_\gamma : \mathcal{A}^m \to \mathcal{A}^{\otimes 2|E(\gamma)|} \otimes \Omega^\bullet(\text{Met}(\gamma))$ as,
    \begin{equation}
    a \mapsto \tilde{K}_\gamma(a) := \otimes_{e \in E(\gamma)} \omega_e \in \mathcal{A}^{\otimes 2|E(\gamma)|} \otimes \Omega^\bullet(\text{Met}(\gamma)) 
    \end{equation}
    Label each vertex $v \in V(\gamma)/[m]$, with a trace map, $Tr_v : \mathcal{A} \to \R$. \\
    Then define the map $K_\gamma : \mathcal{A}^{\otimes m} \to \mathcal{A}^{\otimes n} \otimes \Omega^*(\text{Met}(\gamma))$ as, 
 \begin{equation*}
       a \mapsto K_\gamma (a) := \otimes_{v \in V(\gamma)/[n]} Tr_v (\tilde{K}_\gamma(a))
    \end{equation*}
    \end{Def} 
The tensor order of the output forms is determined by labeling of external vertices.\\
We will use the notation $\tilde{K}_\gamma \in \mathcal{A}^{\otimes 2|E(\gamma)|} \otimes \Omega^\bullet(\text{Met}(\gamma))$ for the form (and not a map) to which $a$ hasn't been multiplied (fed) as input, we use $\tilde{K}_\gamma(-) \in \text{Hom}(\mathcal{A}^m,\mathcal{A}^{\otimes 2|E(\gamma)|} \otimes \Omega^\bullet(\text{Met}(\gamma)))$ as the map which takes in $a$ as an input, and $\tilde{K}_\gamma(a)\in \mathcal{A}^{\otimes 2|E(\gamma)|} \otimes \Omega^\bullet(\text{Met}(\gamma))$ as the final form after taking input $a$.\\
Similarly, we will use the notation $K_\gamma \in \mathcal{A}^{\otimes 2|E(\gamma)|} \otimes \Omega^\bullet(\text{Met}(\gamma)) \otimes \text{Hom}(\mathcal{A}^{2|E(\gamma)|-n})$ for the form (and not a map) that has not yet been multiplied with $a$ as input, we use $\tilde{K}_\gamma(-) \in \text{Hom}(\mathcal{A}^m,\mathcal{A}^{\otimes 2|E(\gamma)|} \otimes \Omega^\bullet(\text{Met}(\gamma)))$ as the map which takes in $a$ as an input, and $\tilde{K}_\gamma(a) \in \mathcal{A}^n \otimes \Omega^\bullet(\text{Met}(\gamma))$ as the final output form.\\
Note that $2|E(\gamma)|=|H(\gamma)|$
\begin{Lemma}
    The map,
    \begin{equation}
        K_\gamma : \mathcal{A}^{\otimes m} \to \mathcal{A}^{\otimes n} \otimes \Omega^*(\text{Met}(\gamma))^{\otimes |E(\gamma)|}
    \end{equation}
    commutes with the differential.
\end{Lemma}
The natural integration pairing between a chain $\sigma$ and a $\text{Met}(\gamma)$-valued form is used to define the following map,

\begin{equation*}
    \sigma \to \int_\sigma K_\gamma (a) 
\end{equation*}

\subsection{Topological conformal field theory} \label{36}
We are almost ready to define TCFT now. We describe both the categories formally and then define a TCFT. For the relevant basic category theory and notation, see Appendix \ref{section category theory background}. 

\begin{Def}
The Segal category of metrized ribbon graphs is a category $\mathscr{C}$ whose objects are non-negative integers, and the set of morphisms $\text{Mor}_\mathscr{C}(m,n)$ consists of all singular cubical chains on the moduli space of all metrized ribbon graphs with labeled m external incoming edges and n external outgoing edges, $\text{Mor}_\mathscr{C}(m,n) := C_\bullet(\Gamma(m,n))$. \footnote {This definition has been originally stated as the Segal Category of Riemann surfaces in \cite{seg} and has been slightly modified here using equivalence of the two categories (explored in \cite{tcft}).} \\
The composition of morphism is defined by gluing outgoing vertices of one graph with incoming vertices of the other such that labeling is respected.       
\end{Def}

\begin{Prop}
The Segal category of metrized ribbon graphs is a category $\mathscr{C}$ can be equipped with the structure of a dgsm category
\end{Prop}

\begin{proof}

\begin{itemize}

\item Define $\sigma_{1,1} \in \text{Mor}_\mathscr{C}(1,1)$ as the 0-chain whose image is the metrized ribbon graph with 1 incoming and 1 outgoing external vertex (labeled by 1 and 1 respectively) connected by an edge of infinitesimal length. Define  another 0-chain $\sigma_{i,i}^{i = 1 \to m} := \sigma_{1,1} \cross \sigma_{2,2} .. \cross \sigma_{m,m} \in \text{Mor}_\mathscr{C}(m,m)$ as the Cartesian product of m-chains like $\sigma_0$ with itself m many times where each connected component of the image graph joining an incoming vertex labeled i with an outgoing vertex labeled i . $\sigma_{i,i}^{i = 1 \to m} \equiv 1_m$ serves as the identity morphism for any integer m. The Cartesian product of chains (thought of as maps between Euclidean spaces) here is assumed to be strictly associative. 
\item The monoidal structure is given by addition of integers and Cartesian product of ribbon graphs. Hence,
$\alpha, \lambda, \rho$ are identity maps on components. Also, Pentagon and triangle identities are trivially satisfied as this is a strict category.
\item We define the braiding on components $R_{m,n} : m+n \to n+m $ as
\begin{equation*}
R_{m,n} := \sigma_{1,n+1} \cross \sigma_{2,n+2} .. \cross \sigma_{m,n+m} \cross \sigma_{m+1,1} \cross \sigma_{m+2,2}..\cross \sigma_{m+n,n}  
\end{equation*}

R clearly squares to identity, and satisfies hexagon identity, since,

\begin{center} \begin{tikzcd}[sep=huge]
  (m + n) + p \arrow[r,"\alpha_{m,n,p}"] \arrow[d,"R_{m,n} \cross 1_p"'] & m + (n + p) \arrow[r,"R_{m,n + p}" ] & (n + p) + m \arrow[d,"\alpha_{n,p,m}"]\\
  (n + m) + p \arrow[r,"\alpha_{n,m,p}"'] &
  n + (m + p) \arrow[r,"1_n \cross R_{m,p}"']  &
  n + (p + m) 
\end{tikzcd}\\
\end{center}

the R.H.S.,
\begin{flalign*}
& 1_n \cross R_{m,p} \cdot \alpha_{n,m,p} \cdot R_{m,n} \cross 1_p = &&\\
& \sigma_{1,1} \cross ..\cross \sigma_{n,n} ..\cross \sigma_{n+1,n+p+1} \cross \sigma_{n+2,n+p+2}..\cross \sigma_{n+m,n+p+m} \cross \sigma_{n+m+1,n+1} \cross \sigma_{m+2,n+2}..\cross \sigma_{m+p,n+p} &&\\ 
& \sigma_{1,n+1} \cross \sigma_{2,n+2}..\cross \sigma_{m,n+m} \cross \sigma_{m+1,1} \cross \sigma_{m+2,2}..\cross \sigma_{m+n,n} \cross \sigma_{m+n+1,n+m+1}..\cross \sigma_{m+n+p,n+m+p} &&\\ 
\end{flalign*}
is equal to the L.H.S.,
\begin{equation*}
   \alpha_{n,p,m} \cdot R_{m,n + p} \cdot \alpha_{m,n,p} = R_{m,n + p} = \sigma_{1,n+p+1} \cross \sigma_{2,n+p+2} .. \cross \sigma_{m,n+p+m} \cross \sigma_{m+1,1} \cross \sigma_{m+2,2}..\cross \sigma_{m+n,n+p}  
    \end{equation*}

\item $C_\bullet(\Gamma(m,n))$ is a chain complex with the standard boundary map on singular chains as the differential $\partial_\mathscr{C}$. 
\end{itemize}
    
\end{proof}

\begin{Def}
    Let Comp denote the category of Chain complexes whose objects are chain complexes (here, graded super vector spaces with a differential), and morphisms are linear maps between the complexes, $\text{Mor}_\text{Comp}(A,B) := \text{Hom}(A,B)$. $\text{Hom}(A,B)$ is a chain complex with the differential as pre-post composition map, $\partial_\text{Comp}(\Phi):=\partial_B \Phi - (-1)^{|\Phi|} \Phi \partial_A$ where $\Phi \in \text{Hom}(A,B)$. The product structure is given by tensor product of chain complexes.
\end{Def}
\begin{Prop}
\label{37}
 Comp is a dgsm category. (
The proof is quite elementary and is provided in the appendix)   
\end{Prop}

\begin{Def}
\label{22}
By \cite{tcft}, an open \textit{topological conformal field theory} (referred to as TCFT throughout this work) is a dgsm functor from $\mathscr{C}$ to Comp. 
\end{Def}
    \begin{Def}
    The data of a Calabi-Yau space defines a map $T$ as follows,
    \begin{equation}
    T(m):=\mathcal{A}^m \hspace{10 pt} ; \hspace{10 pt} T(\sigma) := \int_\sigma K_{\gamma} \in \text{Hom}(\mathcal{A}^m,\mathcal{A}^n)   
    \end{equation}
    where, $m,n$ are integers with $m$ and $n$ many elements respectively; $\sigma$ is a chain (with m incoming and n outgoing edges) in the cell labeled by the unmetrized ribbon graph $\gamma(\sigma) \equiv \gamma$. \\
    \footnote{Technically, throughout this work, almost always the tensor product of spaces actually refers to the completed tensor product of spaces, but we use the same notation ($\overline{A \otimes B} \equiv A \otimes B$). Moreover, the proofs are illustrated only for homogeneous elements, and follow for the whole space by linearity.}
    The following theorem (from \cite{tcft}) is stated here.
\end{Def}
\begin{Th}
   T is a TCFT. 
\end{Th}

\section{Single eigenvalue TCFTs} \label{Single eigenvalue TCFTs}

Now that we have presented the setup from \cite{tcft}, we are ready to present a method for regularizing the TCFT with the ultimate objective of obtaining finite results from possibly divergent integrals that arise when certain restrictions from \cite{tcft} are relaxed.

\subsection{The cutoff heat kernel}

Given the setup of a Calabi-Yau space as in the intro, $(m \text{-dimensional manifold }M,\mathcal{A}=\Omega^\bullet(M,\R),Q=d$, $\text{Tr}_x(a(x)b(x)) = \int_{x \in M} a(x)\wedge b(x), \langle , \rangle, Q^\dagger=d^*, * = \star, H := dd^* + d^* d)$; we present an explicit form of the heat kernel with the idea borrowed from (\cite{book}).\\
We assume (without proof) that the set of eigenvalues of the Laplacian is countable (indexed by $\mathbb{N}$) and the eigenvalues are positive and indexed in an increasing order. For each $i \in \mathbb{N}$, let $\lambda_i$ denote the distinct $i^\text{th}$ eigenvalue and let $e_{ij} \in \Omega^\bullet(M)$ denote the basis vectors of the eigen-space $E_i$ of $\lambda_i$, where $j$ for each $i$ is indexed by at most a countable set $J_i$. Recall the fact that eigenvectors of the Laplacian are orthogonal (and hence linearly independent whenever non-zero).
\begin{Def}
\label{25}
     Define the \textit{complete heat kernel} $K_t(x,y)$ as the following differential form on $M^2$ for any fixed $t \in \R$,
     \begin{equation*}
     K_t(x,y) : = \sum_{i \in \mathbb{N}} e^{-\lambda_i t} \sum_{j \in J_i} e_{ij}(x) \otimes_{\R} \star_y e_{ij}(y) \in \Gamma(M \cross M, \mathcal{A} \boxtimes \mathcal{A} ) \equiv \mathcal{A}^{\otimes 2} ; 
\end{equation*}
As one allows $t \in \R$ to vary, the Heat Kernel is an element $K_t(x,y) \in \Gamma(M \cross M \cross \R_{>0}, \mathcal{A} \boxtimes \mathcal{A} ) \equiv \mathcal{A}^{\otimes 2} \otimes C^\infty(\R_{>0})$
\end{Def}
\begin{Lemma}
     The complete heat kernel $K_t(x,y)$ is a heat kernel for the Laplacian in the sense of definition \ref{31}.
\end{Lemma}
\begin{proof}
    $K_t(x,y)$ clearly satisfies the heat equation, 
    \begin{equation*}
        H_x K_t(x,y) = H_y K_t(x,y) = - \frac{d}{dt} K_t(x,y) = \sum_{i \in \mathbb{N}} \lambda_i e^{-\lambda_i t} \sum_{j \in J_i} e_{ij}(x) \otimes_{\R} \star_y e_{ij}(y)
    \end{equation*}
    and, $ \text{lim}_{t \to 0} K_t(x,y) = \sum_{i \in \mathbb{N}, j \in J_i} e_{ij}(x) \otimes_{\R} \star_y e_{ij}(y) = \delta(x-y)$
\end{proof}
\begin{Cor}
\label{6}
    The complete heat kernel is a closed differential form.
\end{Cor}
\begin{proof}
     The complete heat kernel is a heat kernel for the Laplacian, hence, it is closed by proposition \ref{4}
\end{proof}
\begin{Def}
Let $\lambda_i \in \R_{\geq 0}$ be an eigenvalue of the Laplacian. Define the \textit{eigenvalue $\lambda_i$-heat kernel} $K^\lambda_t(x,y)$ as,
    \begin{equation*}
    K^\lambda_t(x,y) : = e^{-\lambda_i t} \sum_j e_{ij}(x) \otimes_{\R} \star_y e_{ij}(y) \in \Gamma(M \cross M \cross \R_{\geq 0}, \mathcal{A} \boxtimes_{C^\infty(M \cross M \cross \R_{\geq 0})} \mathcal{A} ) \equiv \mathcal{A}^{\otimes 2} \otimes C^\infty(\R_{\geq 0})
        \end{equation*}
\end{Def}
This definition easily extends to \textit{eigenvalue $A$-heat kernel} $K^A_t(x,y)$ over any arbitrary subset $A$ of the set of all eigenvalues of the Laplacian (denoted by $E \subset \R_{\geq0}$) as,
\begin{equation*}
     K^A_t(x,y) : = \sum_{i : \lambda_i \in A} e^{-\lambda_i t} \sum_{j \in J_i} e_{ij}(x) \otimes_{\R} \star_y e_{ij}(y) \in \mathcal{A}^{\otimes 2} \otimes C^\infty(\R_{\geq 0})
\end{equation*}
and in particular becomes the scale infinity heat kernel (Complete heat kernel in our terms) $K_t(x,y)$ when $A = [0,\infty) \cap E$, and becomes the cutoff heat kernel $K^{[0,\Lambda]}_t(x,y) \equiv K^\Lambda_t(x,y)$ used in the usual regularized quantum field theory with an energy cutoff $\Lambda$ when $A = [0,\Lambda] \cap E$. We'll abuse the notation slightly by dropping $\cap E$, and for instance, just writing $A = [0,\Lambda]$, $A= [0,\infty)$, etc.
\begin{Lemma}
\label{18}
    For any fixed $t \in \R_{\geq 0}$, the \textit{eigenvalue $\lambda$ heat kernel} $K^\lambda_t(x,y)$ is a closed differential form on $M \cross M$;
    \begin{equation}
        (d_x + d_y) K^\lambda_t(x,y) = 0
    \end{equation}
\end{Lemma}

\begin{proof}
    By Corollary \ref{6},
    \begin{equation*}
      (d_x + d_y) [  \sum_{i \in \mathbb{N}} e^{-\lambda_i t} \sum_{j \in J_i} e_{ij}(x) \otimes_{\R} \star_y e_{ij}(y) ] = 0
    \end{equation*}
     \begin{equation}
     \label{5}
       \implies  \sum_{i \in \mathbb{N}} e^{-\lambda_i t} (d_x + d_y) [\sum_{j \in J_i} e_{ij}(x) \otimes_{\R} \star_y e_{ij}(y) ] \equiv \sum_{i \in \mathbb{N}} e^{-\lambda_i t} A_i= 0
    \end{equation}

For a fixed $i$,
\begin{equation}
\label{60}
    A_{i} \equiv (d_x + d_y) [\sum_{j \in J_{i}} e_{ij}(x) \otimes_{\R} \star_y e_{ij}(y) ] = \sum_{j \in J_{i}} [d_x e_{ij}(x) \otimes_{\R} \star_y e_{ij}(y) + e_{ij}(x) \otimes_{\R} d_y \star_y e_{ij}(y)] \equiv \sum_{j \in J_{i}} B_{ij} \vspace{5 pt}
\end{equation}
\begin{equation}
    \implies \sum_{i \in \mathbb{N}} e^{-\lambda_i t} \sum_{j \in J_i} B_{ij} = 0 \hspace{30 pt} (\text{by (\ref{5})} \text{and (\ref{60}))}
\end{equation}
Note that, $\forall j \in J_{i} ; B_{ij} \equiv [d_x e_{ij}(x) \otimes_{\R} \star_y e_{ij}(y) + e_{ij}(x) \otimes_{\R} d_y \star_y e_{ij}(y)]$ is an eigenvector of $H_x \otimes 1_y$ with eigenvalue $\lambda_{i}$. This means, $A_{i_l} \perp A_{i_k} ;  i_l, i_k \in \mathbb{N};$ whenever $l \neq k$ (and hence linearly independent whenever non-zero). Hence, by equation (15)
\begin{equation}
  \sum_{i \in \mathbb{N}} e^{-\lambda_i t} A_i= 0 \implies   A_{i} \equiv (d_x + d_y) [\sum_{j \in J_{i}} e_{ij}(x) \otimes_{\R} \star_y e_{ij}(y) ] = 0 ; \forall i \in \mathbb{N} \implies (d_x + d_y) K^\lambda_t(x,y) = 0
\end{equation}
\end{proof}

\begin{Prop}
\label{7}
    We list some other identities (listed below) to further study the cutoff picture.
 
    \begin{equation}
    \label{12}
    (d_x + d_y) K^\Lambda_t(x,y) = 0
\end{equation}
\begin{equation}
\label{13}
     (d^*_x - d^*_y) K^\Lambda_t(x,y) = 0
\end{equation}
   \begin{equation}
   \label{14}
    (H_x - H_y) K^\Lambda_t(x,y) = 0
\end{equation}
 \begin{equation}
 \label{17}
     (d_x+d_y)L^\Lambda_t(x,y) =  - H_x K^\Lambda_t(x,y)=\frac{d}{dt}K^\Lambda_t(x,y)
\end{equation}
\end{Prop}

\begin{proof} 
We have proven (\ref{12}) in Lemma \ref{18}. Proof of (\ref{13}) is quite similar to proof of (\ref{12}).\\ 
Proof of (\ref{14}) : Most proofs are stated just for a homogeneous component (with the same notation) for the sake of dropping the tedious summation notations and indices, linearity ensures proofs for the whole sum. 
\begin{flalign*}
&    H_x K^\Lambda_t(x,y) = (H_x \otimes 1) [  e^{-\lambda t}  e(x) \otimes_{\R} \star_y e(y) ] = [  e^{-\lambda t} \lambda e(x) \otimes_{\R} \star_y e(y) ] = [  e^{-\lambda t}  e(x) \otimes_{\R} \star_y \lambda e(y) ] &&\\
&   = [  e^{-\lambda t}  e(x) \otimes_{\R} \star_y H_y e(y) ] = [ e^{-\lambda t}  e(x) \otimes_{\R} H_y \star_y e(y) ] =  (-1)^{|H_y||e(x)|} (1 \otimes H_y) [ e^{-\lambda t}  e(x) \otimes_{\R} \star_y e(y)] = H_y K^\Lambda_t(x,y) &&\\
\end{flalign*}
since $|H_y|=0$ and $H_y$ and $\star_y$ commute.\\
Proof of (\ref{17}) :
\begin{flalign*}
& (d_x+d_y)L^\Lambda_t(x,y) = - (d_x+d_y) d^*_x K^\Lambda_t(x,y) = - [H_x K^\Lambda_t(x,y) - d^*_x d_x K^\Lambda_t(x,y) + (-1)^{|d^*_x| |d_y|} d^*_x d_yK^\Lambda_t(x,y)]&&\\
& = -H_x K^\Lambda_t(x,y) = \frac{d}{dt}K^\Lambda_t(x,y) &&\\
\end{flalign*}
\end{proof}
\begin{Cor}
\label{9}
All identities in proposition \ref{7} hold true for eigenvalue $A$-heat kernel $K^A_t(x,y)$ (for any countable subset $ A \subset R_{\geq 0}$), simply since the operators are linear. 
\end{Cor}
\subsubsection{Single eigenvalue TCFT}
We slightly modify the labeling procedure by replacing the complete heat kernel in definition \ref{25} by some single eigenvalue heat kernel 
\begin{Def}
    Let $\mathcal{A}_\lambda \subset \mathcal{A}$ denote the eigenspace of $H$ corresponding to the eigenvalue $\lambda$. Since, $H$ and $d$ commute, $d$ preserves the eigenspaces of $H$, and hence $(\mathcal{A}_\lambda,d)$ forms a chain complex (a sub complex of $\mathcal{A}$).\\
    We denote forms in $\mathcal{A}_\lambda$ as 
$a_\lambda$. $a_\lambda \in \mathcal{A}_\lambda$ is a projection for some $a \in \mathcal{A}$
\end{Def}
Let $\gamma$ be an unmetrized ribbon graph with n incoming and m outgoing labeled external vertices. $\text{Met}(\gamma)$ is the space of metrics on $\gamma$. \\ 
Let $a=a_1(x_1) \otimes .. \otimes a_n(x_n)$ where each $a_i(x_i) \in \mathcal{A}$ labels an incoming vertex with respect to labeling of incoming external vertices.

\begin{Def}
Given an edge $e \in E(\gamma)$ joining $v_i,v_j \in V(\gamma)$, we label the edge by the form $\omega^\lambda_e$ which is defined in the following way, \\ 
If $e$ is an internal edge, then,
\begin{equation*}
    \omega^\lambda_e := K^\lambda_t(x_i,x_j) + dt L^\lambda_t(x_i,x_j) \in \mathcal{A}_\lambda^{\otimes 2} \otimes \Omega^*(\text{Met}(\gamma))
\end{equation*}
If $e$ is an incoming external edge with external vertex labeled by $a^\lambda_j(x_j)$, then
\begin{equation*}
    \omega^\lambda_e := [K^\lambda_t(x_1,x_2) + dt L^\lambda_t(x_1,x_2)] \wedge a^\lambda_k(x_k) \in \mathcal{A}_\lambda^{\otimes 2} \otimes \Omega^\bullet(\text{Met}(\gamma))
\end{equation*}
\end{Def}
\begin{Def}
    Define the map $ \tilde{K}^\lambda_\gamma : \mathcal{A}^m_\lambda \to \mathcal{A}_\lambda^{\otimes 2|E(\gamma)|} \otimes \Omega^\bullet(\text{Met}(\gamma))$ as,
    \begin{equation}
    a \mapsto \tilde{K}^\lambda_\gamma(a) := \otimes_{e \in E(\gamma)} \omega_e \in \mathcal{A}_\lambda^{\otimes 2|E(\gamma)|} \otimes \Omega^\bullet(\text{Met}(\gamma))
    \end{equation}
    Label each vertex $v \in V(\gamma)/[m]$, with a trace map, $Tr_v : \mathcal{A}_\lambda \to \R$. \\
    Then define the map $K_\gamma^\lambda : \mathcal{A}_\lambda^{\otimes m} \to \mathcal{A}_\lambda^{\otimes n} \otimes \Omega^*(\text{Met}(\gamma))$ as, 
 \begin{equation*}
       a \mapsto K_\gamma^\lambda (a) := \otimes_{v \in V(\gamma)/[n]} Tr_v (\tilde{K}_\gamma^\lambda(a))
    \end{equation*}
    \end{Def} 
The tensor order of the output forms is determined by labeling of external vertices.
\footnote{We will use the notation $\tilde{K}^\lambda_\gamma \in \mathcal{A}_\lambda^{\otimes 2|E(\gamma)|} \otimes \Omega^\bullet(\text{Met}(\gamma))$ for the form (and not a map) to which $a$ hasn't been multiplied (fed) as input, we use $\tilde{K}^\lambda_\gamma(-) \in \text{Hom}(\mathcal{A}^m,\mathcal{A}_\lambda^{\otimes 2|E(\gamma)|} \otimes \Omega^\bullet(\text{Met}(\gamma)))$ as the map which takes in $a$ as an input, and $\tilde{K}^\lambda_\gamma(a)\in \mathcal{A}_\lambda^{\otimes 2|E(\gamma)|} \otimes \Omega^\bullet(\text{Met}(\gamma))$ as the final form after taking input $a$.\\
Similarly, we will use the notation $K^\lambda_\gamma \in \mathcal{A}_\lambda^{\otimes 2|E(\gamma)|} \otimes \Omega^\bullet(\text{Met}(\gamma)) \otimes \text{Hom}(\mathcal{A}_\lambda^{2|E(\gamma)|-n})$ for the form (and not a map) that has not yet been multiplied with $a$ as input, we use $\tilde{K}^\lambda_\gamma(-) \in \text{Hom}(\mathcal{A}_\lambda^m,\mathcal{A}^{\otimes 2|E(\gamma)|} \otimes \Omega^\bullet(\text{Met}(\gamma)))$ as the map which takes in $a$ as an input, and $\tilde{K}^\lambda_\gamma(a) \in \mathcal{A}_\lambda^n \otimes \Omega^\bullet(\text{Met}(\gamma))$ as the final output form.\\
Note that $2|E(\gamma)|=|H(\gamma)|$}

\begin{Def}
    Define a map $T^\lambda$ from $\mathscr{C} \to \text{Comp}$ as,
    \begin{equation}
        T^\lambda(m) := \Bigl( \mathcal{A}_\lambda^{\otimes m} , d^m \Bigl) \hspace{10 pt} ; \hspace{10 pt} T^\lambda(\sigma) := \int_\sigma K^\lambda_{\gamma} \in \text{Hom}(\mathcal{A}_\lambda^m,\mathcal{A}_\lambda^n)   
    \end{equation}
    where, $m,n$ are integers with $m$ and $n$ many elements respectively; $\sigma$ is a chain (with m incoming and n outgoing edges) in the cell labeled by the unmetrized ribbon graph $\gamma(\sigma) \equiv \gamma$, and \\ $d^m \equiv d \otimes 1 ..\otimes 1 + 1 \otimes d ..\otimes 1 + .. + 1 ..\otimes d $
\end{Def}

\begin{Lemma}
\label{2}
    $T^\lambda : \mathscr{C} \to \text{Comp}$ is a functor. 
\end{Lemma}
\begin{proof}
    In order to prove $T^\lambda$ is a functor, we need to check composition and identity conditions for a functor.\\
    Given any 2 chains $\sigma_1 \in C_*(\Gamma(m,n)),\sigma_2 \in C_*(\Gamma(n,l)), \sigma_{i,i}^{i=1 \to m} \in C_*(\Gamma(m,m))$, we wish to prove, $T^\lambda(\sigma_2 \cdot \sigma_1) (a) = T^\lambda(\sigma_2) \cdot T^\lambda(\sigma_1) (a)$ and $T^\lambda (I_m) (a)= a $; for all $a \in \mathcal{A_\lambda}$ where $\sigma_{i,i}^{i=1 \to m}$ as defined earlier is the identity morphism on $m$\\
    Let $\gamma_1 \in \Gamma^U(l,m) , \gamma_2 \in \Gamma^U(m,n)$ be the unmetrized graphs corresponding to $\sigma_1,\sigma_2$ respectively. \\
    Let $x_i : 1 \leq i \leq (|V(\gamma_1)|-n) $ label any non outgoing vertex of $\gamma_1$, let $y_i : 1 \leq i \leq n $ label any outgoing vertex of $\gamma_1$ and any incoming vertex of $\gamma_2$ such that gluing is respected, let $z_i : 1 \leq i \leq (|V(\gamma_1)|-n) $ label any non incoming vertex of $\gamma_1$.  \\
    Let $r_i : 1 \leq i \leq (|E(\gamma_1)|-n) $ label any non outgoing edge of $\gamma_1$, let $s_i : 1 \leq i \leq n $ label any outgoing edge of $\gamma_1$, let $t_i : 1 \leq i \leq n $ label any incoming edge of $\gamma_2$, let $u_i : 1 \leq i \leq (|E(\gamma_2)|-n) $ label any non incoming edge of $\gamma_2$.\\
    For the sake of being neat, we drop the indices wherever the details are not relevant.\\
    Then,
    \begin{equation*}
        T^\lambda(\sigma_2 \cdot \sigma_1) (a)= \int_{\sigma_2 \cdot \sigma_1} K^\lambda_{\gamma_2 \cdot \gamma_1} (a) = \int_{r,s,t,u} \text{Tr}_{x} a(x) \widetilde{K}^\lambda_{\gamma_1}(x,x) \text{Tr}_{y} \omega_s^\lambda(x,y) \omega_t^\lambda(y,z) \text{Tr}_{z}\widetilde{K}^\lambda_{\gamma_2}(z,z)
    \end{equation*}
    \begin{equation*}
       = \int_{t,u} \bigl( \int_{r,s} \text{Tr}_{x} a(x) \widetilde{K}^\lambda_{\gamma_1}(x,x)  \omega_s^\lambda(x,y) \bigl) \text{Tr}_{y} \omega_t^\lambda(y,z) \text{Tr}_{z}\widetilde{K}^\lambda_{\gamma_2}(z,z) = \int_{t,u} \text{Tr}_{y} \text{Tr}_{z}\omega_t^\lambda(y,z) \widetilde{K}^\lambda_{\gamma_2}(z,z) (K^\lambda_{\gamma_1}(a))(y)
    \end{equation*}
    \begin{equation*}
     = K^\lambda_{\gamma_1} (K^\lambda_{\gamma_1}(a))  = T^\lambda(\sigma_2) \cdot T^\lambda(\sigma_1) (a)
    \end{equation*}
    The identity condition holds as well,
    \begin{equation*}
        T^\lambda (I_m) (a^\lambda) = \prod_{i=1}^{m} \int_{x_i} \underset{t \to 0}{\text{lim}} \omega^\lambda_{t_i} (x_i,y_i) a^\lambda_i(x_i) = \prod_{i=1}^{m} \int_{x_i} \underset{t \to 0}{\text{lim}} K^\lambda_{t_i} (x_i,y_i) a^\lambda_i(x_i) = \prod_{i=1}^{m} \int_{x_i}\delta^\lambda(x_i-z_i)a^\lambda_i(x_i) = a^\lambda_i(z_i) = a^\lambda
    \end{equation*}
 here $\delta^\lambda$ is the projection of the identity distribution $\delta$ on $\mathcal{A}_\lambda$
    \end{proof}
    Note that this proof works for any form labeling the internal and incoming external edges and has nothing to do with the Laplacian, and uses only Fubini's theorem. Lemma \ref{8}, however, depends on the Laplacian and its eigenforms.
 \begin{Lemma}
    \label{26}
   The functor $T^\lambda : \mathscr{C} \to \text{Comp}$ is symmetric monoidal. 
\end{Lemma}
\begin{proof}
    \begin{itemize}

        \item $T^\lambda$ is monoidal by the following, \\
        Let $\mu_{m,n} : \mathcal{A}_\lambda^m \otimes \mathcal{A}_\lambda^n \to \mathcal{A}_\lambda^{m+n}$ be identity on components, $\mu_{m,n} : a \otimes b \mapsto  a \otimes b$. The following diagram commutes as all morphisms are identity morphisms.

        \begin{center} \begin{tikzcd}[sep=huge]
  (\mathcal{A}_\lambda^l \otimes \mathcal{A}_\lambda^m) \otimes \mathcal{A}_\lambda^n \arrow[r,"\beta_{\mathcal{A}_\lambda^l ,\mathcal{A}_\lambda^m,\mathcal{A}_\lambda^n }"] \arrow[d,"\mu_{l,m} \otimes 1_{\mathcal{A}_\lambda^n}"'] & \mathcal{A}_\lambda^l \otimes (\mathcal{A}_\lambda^m \otimes \mathcal{A}_\lambda^n) 
 \arrow[d,"1_{T(l)} \otimes \mu_{m,n} "] \\
\mathcal{A}_\lambda^{l+m} \otimes \mathcal{A}_\lambda^n \arrow[d,"\mu_{l + m, n} "'] & \mathcal{A}_\lambda^l \otimes \mathcal{A}_\lambda^{m+n} \arrow[d,"\mu_{l,m+n}"] \\
 \mathcal{A}_\lambda^{(l+m)+n}
\arrow[r,"T(\alpha_{l,m,n})"'] &
  \mathcal{A}_\lambda^{l+(m+n)}
\end{tikzcd} \\
\end{center}

$T^\lambda(0)=\mathcal{A}_\lambda^0=\R$ and $T^\lambda(1_m)=1_{\mathcal{A}_\lambda^m}$. Define $\epsilon : \R \to \mathcal{A}_\lambda^0 \equiv \R$ as the identity. The following diagram commutes as all morphisms are identity morphisms.

\begin{center} \begin{tikzcd}[sep=huge]
  \R \otimes \mathcal{A}^m_\lambda \arrow[r,"\epsilon \otimes 1_{\mathcal{A}^m_\lambda}"] \arrow[d,"1_{\mathcal{A}_\lambda^m}"'] & \R \otimes \mathcal{A}^m_\lambda \arrow[d,"\mu_{0,m}"] &  \mathcal{A}^m_\lambda \otimes \R \arrow[r,"1_{\mathcal{A}^m_\lambda} \otimes \epsilon"] \arrow[d,"1_{\mathcal{A}_\lambda^m}"'] & \mathcal{A}^m_\lambda \otimes \R \arrow[d,"\mu_{m,0}"]  \\
  \mathcal{A}^m_\lambda  & \mathcal{A}^m_\lambda \arrow[l,"1_{\mathcal{A}_\lambda^m}"] & \mathcal{A}^m_\lambda  & \mathcal{A}^m_\lambda \arrow[l,"1_{\mathcal{A}_\lambda^m}"]
\end{tikzcd}
  \end{center}      

\item T respects braiding (and hence is a symmetric monoidal functor), since the following diagram commutes,

\begin{center} \begin{tikzcd}[sep=huge]
  \mathcal{A}_\lambda^m \otimes \mathcal{A}_\lambda^n \arrow[r,"a \otimes b \mapsto b \otimes a "] \arrow[d," a \otimes b \mapsto a \otimes b "'] & \mathcal{A}_\lambda^n \otimes \mathcal{A}_\lambda^m \arrow[d,"b \otimes a \mapsto b \otimes a"]  \\
  \mathcal{A}_\lambda^{m + n} \arrow[r,"T(R_{A,B})"]  & \mathcal{A}_\lambda^{n + m}
\end{tikzcd}\\
\end{center}

As the morphism $T(R_{A,B})$ maps $a(x_1) .. \otimes .. a(x_m) \otimes b(x_1) .. \otimes .. b(x_n) \mapsto b(z_1) .. \otimes .. b(z_m) \otimes a(z_1) .. \otimes .. a(z_n)$ where $x_i,z_i \in M$

\end{itemize}
\end{proof}

\begin{Lemma}
\label{8}
        The map $T^\lambda : C_\bullet(\Gamma(m,n)) \to \text{Hom}(\mathcal{A}_\lambda^m,\mathcal{A}_\lambda^n)$ that maps $\sigma \mapsto \int_\sigma K_\gamma(-)$ is a chain map.
\end{Lemma}
\begin{proof}
Let the map $\Phi_1 : \text{Hom}( C_\bullet(\Gamma(m,n)) , \text{Hom}(\mathcal{A}_\lambda^m,\mathcal{A}_\lambda^n)) \to \text{Hom}( \mathcal{A}_\lambda^m ,\mathcal{A}_\lambda^n \otimes (C_\bullet(\Gamma(m,n)))^V) $ be the canonical isomorphism : $\text{Hom}(A,\text{Hom}(B,C)) \to \text{Hom}(A,B^V \otimes C) \to \text{Hom}(B,C \otimes A^V)$.\\
The integration pairing  $C_\bullet(\Gamma(m,n)) \otimes \Omega^\bullet(\text{Met}(\gamma)) \to \mathbb{C}$ gives an inclusion chain map $i : \Omega^\bullet(\text{Met}(\gamma))) \to (C_\bullet(\Gamma(m,n)))^V$. This induces a chain map $\Phi_2 : \text{Hom}( \mathcal{A}_\lambda^m ,\mathcal{A}_\lambda^n \otimes (C_\bullet(\Gamma(m,n)))^V) \to \text{Hom}( \mathcal{A}_\lambda^m ,\mathcal{A}_\lambda^n \otimes \Omega^\bullet(\text{Met}(\gamma))) $.\\
Hence, the composition of maps $\Phi_2 \cdot \Phi_1 = \Phi$, is a chain map.
\begin{equation*}
\Phi : \text{Hom}( C_\bullet(\Gamma(m,n)) , \text{Hom}(\mathcal{A}_\lambda^m,\mathcal{A}_\lambda^n)) \to \text{Hom}( \mathcal{A}_\lambda^m ,\mathcal{A}_\lambda^n \otimes (C_\bullet(\Gamma(m,n)))^V) \to \text{Hom}( \mathcal{A}_\lambda^m ,\mathcal{A}_\lambda^n \otimes \Omega^\bullet(\text{Met}(\gamma)))
\end{equation*}
 By lemma \ref{21}, the map $K^\lambda_{\gamma}(-) \in \text{Hom}(\mathcal{A}_\lambda^m,\mathcal{A}_\lambda^n \otimes \Omega^\bullet(\text{Met}(\gamma)))$ that takes $a \mapsto K^\lambda_\gamma(a) $ is a chain map. Hence, $\Phi(K^\lambda_{\gamma}(-))=T^\lambda$ is a chain map.\\
\end{proof}  
\begin{Lemma}
\label{20}
 $\omega_e \in \mathcal{A}_\lambda^{\otimes 2} \otimes \Omega^*(\R_{\geq 0})$ is a closed form, i.e., $(d_x+d_y+d_t)(K_t(x,y) + d_t t L_t(x,y)) = 0$   \end{Lemma}
\begin{proof}
\begin{flalign*}
& (d_x+d_y+d_t)(K_t^\lambda(x,y) + d_t t L_t^\lambda(x,y)) = (d_x+d_y)K_t^\lambda(x,y) +d_t K_t^\lambda(x,y)+ (d_x+d_y)d_t t L_t^\lambda(x,y)  +d_t d_t t L_t^\lambda(x,y) &&\\  
& = 0+d_t K_t^\lambda(x,y) + (d_x+d_y) d_t t L_t^\lambda(x,y)  + 0 = d_t t \frac{d}{dt}K_t^\lambda(x,y) + (-1)^{|d_t t||d_x+d_y|} d_t t (d_x+d_y) L_t^\lambda(x,y) &&\\ 
& = d_t t \frac{d}{dt}K_t^\lambda(x,y)  - d_t t \frac{d}{dt}K_t^\lambda(x,y) = 0 &&\\
\end{flalign*}
\end{proof}
\begin{Lemma} 
\label{21}
The map $ K^\lambda_{\gamma}(-) \in \text{Hom}(\mathcal{A}_\lambda^m,\mathcal{A}_\lambda^n \otimes \Omega^\bullet(\text{Met}(\gamma)))$ that takes $a \mapsto K^\lambda_\gamma(a) $ is a chain map, 
\end{Lemma}
\begin{proof} 
\begin{itemize}
\item By lemma, $\omega_{e_i} \in \mathcal{A}_\lambda^2 \otimes \Omega^*(\R_{\geq 0}) \equiv \Omega^\bullet(M)^{\otimes 2} \otimes \Omega^\bullet(\R_{\geq 0}) \equiv \Omega^\bullet(M^2 \cross \R_{\geq 0})$ is closed for all edges $e_i$. We denote the projection map as
\begin{equation*}
\pi_{e_i} : M^{2|E(\gamma)|} \cross \R_{\geq 0}^{|E(\gamma)|} \to M^2 \cross \R_{\geq 0}
\end{equation*}
and its pullback on forms as,
\begin{equation*}
\pi^*_{e_i} :  \Omega^\bullet(M)^{\otimes 2} \otimes \Omega^\bullet(\R_{\geq 0}) \to \Omega^\bullet(M)^{\otimes 2|E(\gamma)|} \otimes \Omega^\bullet(\text{Met}(\gamma))
\end{equation*}
Hence, $\pi^*_{e_i} \omega^\lambda_{e_i} \in \Omega^\bullet(M)^{\otimes 2|E(\gamma)|} \otimes \Omega^\bullet(\text{Met}(\gamma))$ is a closed form, as pullback commutes with differential.\\
$\underset{i}{\otimes} \pi^*_{e_i} \omega^\lambda_{e_i} \equiv \widetilde{K^\lambda_\gamma} \in \Omega^\bullet(M)^{\otimes 2|E(\gamma)|} \otimes \Omega^\bullet(\text{Met}(\gamma))$ is closed, as tensor product of closed forms is closed.
\item The map $T^* : \mathcal{A}_\lambda^m \to  \mathcal{A}_\lambda^m \otimes \mathcal{A}_\lambda^{\otimes 2|E(\gamma)|} \otimes \Omega^\bullet(\text{Met}(\gamma)) $ that takes $ a \to a \otimes 1 \mapsto a \otimes \widetilde{K_\gamma}$ is a tensor product of 2 chain maps and is therefore, a chain map.
\item Now, denote the diagonal map (for any single pair of edges sharing a common vertex) as $D_{e_1,e_2} : M^3 \to M^4 $ that takes $(x,y,z) \mapsto (x,y,y,z) $ and it's pullback on forms ($D_{e_1,e_2}^* : (\Omega^\bullet(M))^{\otimes 4} \to (\Omega^\bullet(M))^{\otimes 3} $) is given by,
\begin{equation*}
D_{e_1,e_2}^*[(\alpha_1(x) \otimes \beta_1(y)) \otimes (\alpha_2(y) \otimes \beta_2(z))] = \alpha_1(x) \otimes \beta_1(y) \wedge \alpha_2(y) \otimes \beta_2(z)
\end{equation*}
Denote it's pullback on product of forms of all edges as $D^* : \mathcal{A}_\lambda^m \otimes \mathcal{A}_\lambda^{\otimes 2|E(\gamma)|} \otimes \Omega^\bullet(\text{Met}(\gamma)) \to \mathcal{A}_\lambda^{\otimes |V(\gamma)|} \otimes \Omega^\bullet(\text{Met}(\gamma))$. The pullback form $D^* (\underset{i}{\otimes} \pi^*_{e_i} \omega_{e_i} \underset{j}{\otimes} a_j ) \equiv 
D^*(a \otimes \widetilde{K_\gamma}) \in \Omega^\bullet(M^{|V(\gamma)|}) \otimes \Omega^\bullet(\text{Met}(\gamma))$ is closed, (again, since pullback commutes with differential) and $D^*$ is a chain map.
\item Finally, integration along the fiber (over all ($|V(\gamma)|-n$) non outgoing vertices) makes the map, $I^* : \mathcal{A}_\lambda^{\otimes |V(\gamma)|} \otimes \Omega^\bullet(\text{Met}(\gamma)) \to \mathcal{A}_\lambda^n \otimes \Omega^\bullet(\text{Met}(\gamma))$ a chain map.\\
As a result, the composition $K = I^* \cdot D^* \cdot T^* : \mathcal{A}_\lambda^m \to\mathcal{A}_\lambda^n \otimes \Omega^\bullet(\text{Met}(\gamma)) $ that takes $a \mapsto a \otimes \widetilde{K_\gamma} \mapsto D^*(a \otimes \widetilde{K_\gamma}) \mapsto K_\gamma(a) $ is a chain map.
\end{itemize}
\end{proof}

\begin{Th}
    $T^\lambda$ is a TCFT.
\end{Th}

\begin{proof}
    The theorem readily follows from the lemmas \ref{2}, \ref{26} and \ref{8}.
\end{proof}

\section{Constructing TCFTs using single-eigenvalue TCFTs} \label{Constructing TCFTs using single-eigenvalue TCFTs}
We wish to construct a Calabi-Yau TCFT by adding 2 Calabi-Yau TCFTs. We define an idea of addition, and explore if it results in a well defined dgsm functor.
\subsection{Adding two TCFTs}
\begin{Def}
\label{28}
     Let $T^{\alpha}$ and $T^{\beta}$ be 2 eigenvalue TCFTs for $\alpha, \beta \in [0,\infty)$. We assume $\alpha \neq \beta$ throughout this section\\
     Define the sum $T^{\alpha} + T^{\beta} : \mathscr{C} \to \text{Comp}$ to be the map,
 \begin{equation}
     (T^{\alpha} + T^{\beta}) (m) := 
\Bigl( (\mathcal{A}_\alpha \oplus \mathcal{A}_\beta)^{\otimes m} , d^m \Bigl) \equiv \mathcal{A}^m_{\alpha,\beta} \hspace{10 pt} ; \hspace{10 pt} (T^{\alpha} + T^{\beta})(\sigma) := T^{\alpha}(\sigma) + T^{\beta}(\sigma)
 \end{equation}
     \end{Def}
     The following lemma and corollary are used to prove lemma \ref{38}
     \begin{Lemma} 
     \label{27}
     
     \begin{equation}
         \text{Tr}_{y} \omega_s^\alpha(x,y) \wedge \omega_t^\beta(y,z) = 0
        \end{equation}
     \end{Lemma}
     \begin{proof}
\begin{flalign*}
& \text{Tr}_{y} \omega_s^\alpha(x,y) \wedge \omega_t^\beta(y,z)] = \text{Tr}_{y} ( K^\alpha_s(x,y) - ds Q^\dagger_x K^\alpha_s(x,y)) ( K^\beta_t(y,z) - dt Q^\dagger_y K^\beta_t(y,z))  && \\ 
& = \text{Tr}_{y} (K^\alpha_s(x,y) - ds d_x^* K^\alpha_s(x,y))(K^\beta_t(y,z) - dt d_y^* K^\beta_t(y,z)) && \\ 
& = \text{Tr}_{y} (e^{-s\alpha} \sum_{j \in J} e_{ij}(x) \otimes \star e_{ij}(y))(e^{-t\beta} \sum_{l \in L} e_{kl}(y) \otimes \star e_{kl}(z)) && \\
& - dt \text{Tr}_{y} (e^{-s\alpha} \sum_{j \in J} e_{ij}(x) \otimes \star e_{ij}(y))(e^{-t\beta} \sum_{l \in L} d^*_y e_{kl}(y) \otimes \star e_{kl}(z)) && \\
& - ds \text{Tr}_{y} (e^{-s\alpha} \sum_{j \in J} d^*_x e_{ij}(x) \otimes \star e_{ij}(y))(e^{-t\beta} \sum_{l \in L} e_{kl}(y) \otimes \star e_{kl}(z)) && \\
& + dsdt \text{Tr}_{y} (e^{-r\alpha} \sum_{j \in J} d^*_x e_{ij}(x) \otimes \star e_{ij}(y))(e^{-r\beta} \sum_{l \in L} d^*_y e_{kl}(y) \otimes \star e_{kl}(z)) \biggl) && \\
& = 0. \hspace{10 pt} \text{since, for any j,l;} &&\\
& \text{the first term,} &&\\
& \text{Tr}_{y}  (e_{ij}(x) \otimes \star e_{ij}(y)) (e_{kl}(y) \otimes \star e_{kl}(z)) = (-1)^{|\text{Tr}_{y}| |e_{ij}(x)|} e_{ij}(x) \text{Tr}_{y}  [e_{kl}(y) \star e_{ij}(y)] \otimes \star e_{kl}(z) =0 && \\
& ( \text{since,}\hspace{5 pt} \text{Tr}_{y} [e_{kl}(y) \star e_{ij}(y)] = \langle e_{kl}(y) , e_{ij}(y) \rangle =0; \hspace{2 pt} \text{because eigenspaces of Laplacian are orthogonal}) && \\ 
& \text{the second term;} && \\
& \text{Tr}_{y} (e_{ij}(x) \otimes \star e_{ij}(y)) (d^*_y e_{kl}(y) \otimes \star e_{kl}(z)) = (-1)^{|\text{Tr}_{y}| |e_{ij}(x)|} e_{ij}(x) \text{Tr}_{y} [d^*_y e_{kl}(y) \star e_{ij}(y)] \otimes \star e_{kl}(z)) = 0;&& \\
& (\text{since, $d^*_y e_{kl}(y)$ is again in the eigenspace of $\beta$, because $d^*_y$ preserves eigenspaces of H as $[d^*_y,H]=0$)}. && \\
& \text{Similarly the third and fourth terms vanish as well.} && \\
\end{flalign*}
\end{proof}
\begin{Cor}
    \begin{equation*}
         T^{\alpha}(\sigma_2) \cdot T^{\beta}(\sigma_1) (a) = 0
    \end{equation*}
\end{Cor}
\begin{proof}
      \begin{equation*}
        T^{\alpha}(\sigma_2) \cdot T^{\beta}(\sigma_1) (a) = \int_{r,s,t,u} \text{Tr}_{x} a^\beta(x) \widetilde{K}^\beta_{\gamma_1}(x,x)\text{Tr}_{y} \omega_s^\beta(x,y) \omega_t^\alpha(y,z) \text{Tr}_{z}\widetilde{K}^\alpha_{\gamma_2}(z,z) = 0
    \end{equation*}
\end{proof}
\begin{Lemma} 
\label{38}
    $T^{\alpha} + T^{\beta} : \mathscr{C} \to \text{Comp}$ is a functor.
\end{Lemma}
\begin{proof}
    We wish to prove, $(T^{\alpha} + T^{\beta})(\sigma_2 \cdot \sigma_1) (a) = (T^{\alpha} + T^{\beta})(\sigma_2) \cdot (T^{\alpha} + T^{\beta})(\sigma_1) (a)$; $ \forall a \in \mathcal{A}^m_{\alpha,\beta}$ and for all $\sigma_1 \in C_\bullet(\Gamma(m,n)), \sigma_2 \in C_\bullet(\Gamma(n,l))$\\
    Let $\gamma_1 \in \Gamma^U(m,n) , \gamma_2 \in \Gamma^U(n,l)$ be the unmetrized graphs corresponding to $\sigma_1,\sigma_2$ respectively. Then,
    \begin{flalign*}
       & (T^{\alpha} + T^{\beta})(\sigma_2) \cdot (T^{\alpha} + T^{\beta})(\sigma_1) (a) = (T^{\alpha}(\sigma_2) + T^{\beta}(\sigma_2)) \cdot (T^{\alpha}(\sigma_1) + T^{\beta}(\sigma_1)) (a) && \\
            &      = T^{\alpha}(\sigma_2) \cdot T^{\alpha}(\sigma_1) (a) + T^{\alpha}(\sigma_2) \cdot T^{\beta}(\sigma_1) (a) + T^{\beta}(\sigma_2) \cdot T^{\alpha}(\sigma_1) (a) + T^{\beta}(\sigma_2) \cdot T^{\beta}(\sigma_1) (a) && \\
     &   =   T^{\alpha}(\sigma_2 \cdot \sigma_1) (a) +0 + 0 + T^{\beta}(\sigma_2 \cdot \sigma_1) (a) =      (T^{\alpha} + T^{\beta} ) (\sigma_2 \cdot \sigma_1)(a) && \\
    \end{flalign*}
\end{proof}
\begin{Lemma}
\label{3}
    The  functor $T^{\alpha} + T^{\beta} : \mathscr{C} \to \text{Comp}$ is symmetric monoidal.
\end{Lemma}
\begin{proof}
\begin{itemize}
Since,     \begin{equation*}
        (T^\alpha + T^\beta)(m) \otimes (T^\alpha + T^\beta)(n) = ((\mathcal{A}_\alpha \oplus \mathcal{A}_\beta) )^m \otimes ((\mathcal{A}_\alpha \oplus \mathcal{A}_\beta) )^n = ((\mathcal{A}_\alpha \oplus \mathcal{A}_\beta) )^{m+n} = (T^\alpha + T^\beta)(m + n),
    \end{equation*}
           \item T is monoidal by, \\
        Defining $\mu_{m,n} : (\mathcal{A}_\alpha \oplus \mathcal{A}_\beta) ^m \otimes (\mathcal{A}_\alpha \oplus \mathcal{A}_\beta) ^n \to (\mathcal{A}_\alpha \oplus \mathcal{A}_\beta) ^{m+n}$ to be the identity on components, $\mu_{m,n} : a \otimes b \mapsto  a \otimes b$. The following diagram commutes as all morphisms are identity morphisms.
\begin{center} \begin{tikzcd}[sep=huge]
  (\mathcal{A}_{\alpha,\beta}^l \otimes \mathcal{A}_{\alpha,\beta}^m) \otimes \mathcal{A}_{\alpha,\beta}^n \arrow[r,"\beta_{\mathcal{A}_{\alpha,\beta}^l ,\mathcal{A}_{\alpha,\beta} ^m,\mathcal{A}_{\alpha,\beta}^n }"] \arrow[d,"\mu_{l,m} \otimes 1_{\mathcal{A}_{\alpha,\beta}^n}"'] & \mathcal{A}_{\alpha,\beta}^l \otimes (\mathcal{A}_{\alpha,\beta}^m \otimes \mathcal{A}_{\alpha,\beta}^n) 
\arrow[d,"1_{T(l)} \otimes \mu_{m,n} "] \\
\mathcal{A}_{\alpha,\beta}^{l+m} \otimes \mathcal{A}_{\alpha,\beta}^n \arrow[d,"\mu_{l + m, n} "'] & \mathcal{A}_{\alpha,\beta}^l \otimes \mathcal{A}_{\alpha,\beta}^{m+n} \arrow[d,"\mu_{l,m+n}"] \\
\mathcal{A}_{\alpha,\beta}^{(l+m)+n}
\arrow[r,"T(\alpha_{l,m,n})"'] &
\mathcal{A}_{\alpha,\beta}^{l+(m+n)}
\end{tikzcd} \\
\end{center}

$T^{\alpha,\beta}(0)=\mathcal{A}_{\alpha,\beta}^0=\R$ and $T^{\alpha,\beta}(1_m)=1_{\mathcal{A}_{\alpha,\beta}^m}$. Define $\epsilon : \R \to \mathcal{A}_{\alpha,\beta}^0 \equiv \R$ as the identity. The following diagram commutes as all morphisms are identity morphisms.

\begin{center} \begin{tikzcd}[sep=huge]
  \R \otimes \mathcal{A}^m_{\alpha,\beta} \arrow[r,"\epsilon \otimes 1_{\mathcal{A}^m_{\alpha,\beta}}"] \arrow[d,"1_{\mathcal{A}_{\alpha,\beta}^m}"'] & \R \otimes \mathcal{A}^m_{\alpha,\beta} \arrow[d,"\mu_{0,m}"] &  \mathcal{A}^m_{\alpha,\beta} \otimes \R \arrow[r,"1_{\mathcal{A}^m_{\alpha,\beta}} \otimes \epsilon"] \arrow[d,"1_{\mathcal{A}_{\alpha,\beta}^m}"'] & \mathcal{A}^m_{\alpha,\beta} \otimes \R \arrow[d,"\mu_{m,0}"]  \\
  \mathcal{A}^m_{\alpha,\beta}  & \mathcal{A}^m_{\alpha,\beta} \arrow[l,"1_{\mathcal{A}_{\alpha,\beta}^m}"] & \mathcal{A}^m_{\alpha,\beta}  & \mathcal{A}^m_{\alpha,\beta} \arrow[l,"1_{\mathcal{A}_{\alpha,\beta}^m}"]
\end{tikzcd}
 \end{center}

\item T respects braiding (and hence is a symmetric monoidal functor), since the following diagram commutes,

\begin{center} \begin{tikzcd}[sep=huge]
  \mathcal{A}_{\alpha,\beta} ^m \otimes \mathcal{A}_{\alpha,\beta} ^n \arrow[r,"a \otimes b \mapsto b \otimes a "] \arrow[d," a \otimes b \mapsto a \otimes b "'] & \mathcal{A}_{\alpha,\beta} ^n \otimes \mathcal{A}_{\alpha,\beta} ^m \arrow[d,"b \otimes a \mapsto b \otimes a"]  \\
  \mathcal{A}_{\alpha,\beta} ^{m + n} \arrow[r,"T(R_{A,B})"]  & \mathcal{A}_{\alpha,\beta} ^{n + m}
\end{tikzcd}\\
\end{center}

As the morphism $T(R_{A,B})$ maps $a(x_1) .. \otimes .. a(x_m) \otimes b(x_1) .. \otimes .. b(x_n) \mapsto b(z_1) .. \otimes .. b(z_m) \otimes a(z_1) .. \otimes .. a(z_n)$ where $x_i,z_i \in M$

\end{itemize}
\end{proof}
\begin{Lemma}
\label{29}
    $T^{\alpha} + T^{\beta} : C_\bullet(\Gamma(m,n)) \to \text{Hom}(\mathcal{A}_{\alpha,\beta}^m,\mathcal{A}_{\alpha,\beta}^n) $ is a chain map
\end{Lemma}
\label{30}
\begin{proof} Since, $T^{\alpha}, T^{\beta}$ are chain maps,
   \begin{equation*}
    (T^{\alpha} + T^{\beta}) \partial_C (\sigma) - (-1)^p d_{\text{Comp}} (T^{\alpha} + T^{\beta}) (\sigma) = [T^{\alpha} \partial_C (\sigma) - (-1)^p d_{\text{Comp}} T^{\alpha}(\sigma)] + [T^{\beta} \partial_C (\sigma) - (-1)^p d_{\text{Comp}} T^{\beta}(\sigma)] = 0
\end{equation*} 
where, $p \equiv |T^\alpha|=|T^\beta|=|T^\alpha+T^\beta|$ . Hence, $T^{\alpha} + T^{\beta}$ is also a chain map.
\end{proof}
\begin{Th}
     $T^{\alpha} + T^{\beta}$ is a TCFT if $\alpha \neq \beta$
\end{Th}
\begin{proof}
    The proof readily follows from the lemmas \ref{38}, \ref{3} and \ref{29}.
\end{proof}
\begin{Prop}
    The sum of 2 TCFT's defined as in \ref{28}, is not a functor in general.
\end{Prop}
\begin{proof}
     Consider the map $T_\alpha+T_\alpha : C_\bullet(\Gamma(m,n)) \to \text{Hom}(\mathcal{A}_{\alpha,\beta}^m,\mathcal{A}_{\alpha,\beta}^n) $,\\
     Let $\sigma_1 \in C_\bullet(\Gamma(m,n)), \sigma_2 \in C_\bullet(\Gamma(n,l))$ and let $a \in \mathcal{A}_{\alpha,\beta}^m$\\
     We wish to prove, $(T^{\alpha} + T^{\alpha})(\sigma_2 \cdot \sigma_1) (a) \neq (T^{\alpha} + T^{\alpha})(\sigma_2) \cdot (T^{\alpha} + T^{\alpha})(\sigma_1) (a)$ \\
    Let $\gamma_1 \in \Gamma^U(l,m) , \gamma_2 \in \Gamma^U(m,n)$ be the unmetrized graphs corresponding to $\sigma_1,\sigma_2$ respectively. Then,
    \begin{equation*}
        (T^{\alpha} + T^{\alpha})(\sigma_2) \cdot (T^{\alpha} + T^{\alpha})(\sigma_1) (a) = (T^{\alpha}(\sigma_2) + T^{\alpha}(\sigma_2)) \cdot (T^{\alpha}(\sigma_1) + T^{\alpha}(\sigma_1)) (a) 
            \end{equation*}
            \begin{equation*}
        = T^{\alpha}(\sigma_2) \cdot T^{\alpha}(\sigma_1) (a) + T^{\alpha}(\sigma_2) \cdot T^{\alpha}(\sigma_1) (a) + T^{\alpha}(\sigma_2) \cdot T^{\alpha}(\sigma_1) (a) + T^{\alpha}(\sigma_2) \cdot T^{\alpha}(\sigma_1) (a) 
    \end{equation*}
    \begin{equation*}
    =   T^{\alpha}(\sigma_2 \cdot \sigma_1) (a) + T^{\alpha}(\sigma_2 \cdot \sigma_1) (a) + T^{\alpha}(\sigma_2 \cdot \sigma_1) (a) + T^{\alpha}(\sigma_2 \cdot \sigma_1) (a) =     2 (T^{\alpha} + T^{\alpha} ) (\sigma_2 \cdot \sigma_1)(a) \neq (T^{\alpha} + T^{\alpha} ) (\sigma_2 \cdot \sigma_1)(a)
    \end{equation*}
\end{proof}
\subsection{Cutoff TCFT}
We now describe the an alternate way to define a Calabi-Yau TCFT using eigenvalue A-heat Kernel $K^A_t(x,y)$ (definition 2.4), which is more reminiscent of a usual Quantum field theory cutoff. We define this as follows,
\begin{Def}
\label{40}
Let $A$ be a countable subset of $[0,\infty)$. Let the map
      $ T^A : \mathscr{C} \to \text{Comp}$ be defined as,
       \begin{equation}
     T^A (m) := \underset{\lambda \in A}{\oplus} \mathcal{A}^m_\lambda \equiv \mathcal{A}^m_A \hspace{10 pt} ; \hspace{10 pt} T^A(\sigma) :=   \underset{\lambda \in A}{\sum} T^{\lambda}(\sigma) = \int_\sigma K^A_\gamma
 \end{equation}
\end{Def}

\begin{Th}
$T^A$ is a TCFT
\end{Th}
\begin{proof}
  \begin{itemize}
      \item $T^A$ is a functor (Proof is simlar to proof of lemma \ref{2}).
\item $T^A$ is symmetric monoidal by (Proof is simlar to proof of lemma \ref{3}).
\item $K^A_t(x,y)$ satisfies the properties in Corollary \ref{9}, hence, lemma \ref{20} (and consecutively lemma \ref{21}) hold for $K^A_t(x,y)$. Then by lemma \ref{8}, $T^A$ is a chain map. 
  \end{itemize}  
\end{proof}
The following proposition shows that the TCFTs obtained by definition \ref{28} definition \ref{40} are identical.
\begin{Prop}
\label{1}
    \begin{equation}
        (T^{\alpha} + T^{\beta})(\sigma) = \int_\sigma K^\alpha_\gamma + K^\beta_\gamma = \int_\sigma K^{\alpha,\beta}_\gamma = T^{\alpha,\beta}(\sigma) \hspace{10 pt} \forall \sigma \in C_\bullet(\Gamma(m,n))
    \end{equation} 
\end{Prop}
\begin{proof}
\begin{flalign*}
& \int_\sigma K^{\alpha,\beta}_\gamma = \int_{r,s,t,u} \text{Tr}_{x} (a^\alpha(x) + a^\beta(x)) (\widetilde{K}^\alpha_{\gamma_1}(x,x) + \widetilde{K}^\beta_{\gamma_1}(x,x))\text{Tr}_{y} (\omega_s^\alpha(x,y) + \omega_s^\beta(x,y)) (\omega_t^\alpha(y,z) + \omega_t^\beta(y,z)) &&\\
& \text{Tr}_{z}(\widetilde{K}^\alpha_{\gamma_2}(z,z) + \widetilde{K}^\beta_{\gamma_2}(z,z)) &&\\
& = \int_{r,s,t,u} \text{Tr}_{x} a^\alpha(x)\widetilde{K}^\alpha_{\gamma_1}(x,x)\text{Tr}_{y} \omega_s^\alpha(x,y) \omega_t^\alpha(y,z) \text{Tr}_{z}\widetilde{K}^\alpha_{\gamma_2}(z,z) &&\\
& + \int_{r,s,t,u} \text{Tr}_{x} a^\beta(x)\widetilde{K}^\beta_{\gamma_1}(x,x)\text{Tr}_{y} \omega_s^\beta(x,y) \omega_t^\beta(y,z)\text{Tr}_{z}\widetilde{K}^\beta_{\gamma_2}(z,z) &&\\
&   \text{(Cross terms vanish, by lemma \ref{27})} &&\\
& = \int_\sigma K^\alpha_\gamma + K^\beta_\gamma = (T^{\alpha} + T^{\beta})(\sigma) &&\\
\end{flalign*} 
\end{proof}
TCFT constructed in this section is how one usually regulates a QFT by a sharp energy cutoff. We now relate any two regularized TCFTs with different cutoffs. 
\begin{Th}
    Let $\Lambda' > \Lambda > 0$. Then, \\
    $T^{[0,\Lambda]}$ is a TCFT (called as energy eigenvalue $\Lambda$ cutoff TCFTs),
    \begin{equation}
        T^{[0,\Lambda']} = T^{[0,\Lambda]} + T^{[\Lambda,\Lambda']} 
    \end{equation}
\end{Th}
\begin{proof}
    The proof readily follows from Proposition \ref{1} as, 
     \begin{equation}
        T^{[0,\Lambda']} = \sum_{i=0}^{i_{\Lambda'}} T^{\lambda_i} = \sum_{i=0}^{i_{\Lambda}} T^{\lambda_i} + \sum_{i=i_{\Lambda}}^{i_{\Lambda'}} T^{\lambda_i} = T^{[0,\Lambda]} + T^{[\Lambda,\Lambda']} 
    \end{equation}
\end{proof}

To conclude, we presented a way to construct eigenvalue cutoff Heat Kernel and construct regularized TCFTs. In our next work, we would explore length based cutoffs like in quantum field theory and attempt to construct a regularized TCFT which is cutoff independent using either of the regularization methods. That will allow us to extend the TCFT on the compactified moduli space.

\appendix

\section{Category Theory Background} \label{section category theory background}

Definitions \ref{23} to \ref{24} are all adapted from \cite{mc}

\begin{Def}
\label{23}
    A \textit{monoidal category} $(\mathcal{C}, \square, I)$ is a category $\mathcal{C}$ with the following structure,
    \begin{enumerate}
        \item a bifunctor $\square : \mathcal{C} \cross \mathcal{C} \to \mathcal{C}$, called the \textit{tensor product}. Here $\mathcal{C} \cross \mathcal{C}$ is the product category of $\mathcal{C}$ with itself. Denote $\square(A,B) \equiv
        A \square B$ and $\square(f,g) \equiv
        f \square g$ for all $A,B \in \text{Ob}(\mathcal{C})$ and $f,g \in \text{Mor}_\mathcal{C}(-,-)$
        \item an object $I \in \text{Ob}(\mathcal{C})$ called the unit object or \textit{tensor unit}.
        \item Three natural isomorphisms,
        \begin{equation*}
         \alpha : (- \square -) \square - \to - \square (- \square -) \hspace{10 pt} \text{called the \textit{associator}; with components,}
        \end{equation*}
        \begin{equation*}
             \alpha_{A,B,C} : (A \square B) \square C \to A \square (B \square C) 
        \end{equation*}
        \begin{equation*}
         \lambda : I \square - \to - \hspace{10 pt} \text{called the \textit{left unitor} ; with components,}
        \end{equation*}
        \begin{equation*}
             \lambda_{A} : I \square A \to A  
        \end{equation*}
        \begin{equation*}
         \rho : - \square I \to - \hspace{10 pt} \text{called the \textit{right unitor} ; with components,}
        \end{equation*}
        \begin{equation*}
             \rho_{A} : A \square I \to A  
        \end{equation*}
        \item Two conditions on the natural isomorphisms,
       
        \begin{itemize}
            \item the \textit{pentagon identity} which states that the following diagram should commute for all objects,
            
            \begin{center} \begin{tikzcd}[sep=huge] 
            ((A \square B) \square C) \square D \arrow[r,"\alpha_{A \square B,C,D}"] \arrow[d,"\alpha_{A,B,C} \square 1_D"'] & (A \square B) \square (C \square D) \arrow[r,"\alpha_{A,B,C \square D}" ] & A \square (B \square (C \square D)) \\
  (A \square (B \square C)) \square D \arrow[r,"\alpha_{A,B \square C,D}"'] &
  A \square ((B \square C) \square D) \arrow[ur,"1_A \square \alpha_{B,C,D}"']
\end{tikzcd}
\end{center}

            \item the \textit{triangle identity} which states that the following diagram should commute for all objects,

 \begin{center} \begin{tikzcd}[sep=huge]
  (A \square I) \square B \arrow[r,"\alpha_{A ,I,B}"] \arrow[d,"\rho_{A} \square 1_B"'] & A \square (I \square B) \arrow[dl,"1_A \square \lambda_B "]  \\
  A \square B
\end{tikzcd}
 \end{center}           
        \end{itemize}
    \end{enumerate}
\end{Def}

\begin{Def}
Let $(\mathcal{C}, \square, I)$ be a monoidal category. Using the monoidal bifunctor define the \textit{twist bifunctor} $\widetilde{\square} : \mathcal{C} \cross \mathcal{C} \to \mathcal{C}$ as $(A,B) \mapsto B \square A$.\\
A natural isomorphism $S$ between the monoidal functor $\square$ and the twist functor $\widetilde{\square}$ is called \textit{braiding} $S_{A,B} : A \square B \to B \square A $, if the following diagram commutes. (There's an additional similar diagram with inverses that needs to commute, but since we only care about symmetric monoidal categories and functors, that diagram commutes if and only if the one we have commutes) (\textit{hexagon identity}),
       
\begin{center}
     \begin{tikzcd}[sep=huge]
  (A \square B) \square C \arrow[r,"\alpha_{A,B,C}"] \arrow[d,"S_{A,B} \square 1_C"'] & A \square (B \square C)\arrow[r,"S_{A,B \square C}" ] & (B \square C) \square A \arrow[d,"\alpha_{B,C,A}"]\\
  (B \square A) \square C \arrow[r,"\alpha_{B,A,C}"'] &
  B \square (A \square C) \arrow[r,"1_B \square S_{A,C}"']  &
  B \square (C \square A) 
\end{tikzcd}\\
\end{center}
       
A \textit{braided monoidal category} is a monoidal category $(\mathcal{C}, \square, I)$ equipped with a braiding $S$.\\
If in addition to the hexagon identity, $S$ satisfies $S_{B,A} \cdot S_{A,B} = I_{A \square B}$ for all objcts/morphisms $A,B \in \mathcal{C}$, then $(\mathcal{C}, \square, I, S)$ is said to be a \textit{symmetric monoidal category}.
\end{Def}

\begin{Def}
    A \textit{differential graded symmetric monoidal (dgsm) category} (\cite{seg}) is a category $\mathscr{C}$ such that,
    \begin{itemize}
    \item $\mathcal{C}$ is a symmetric monoidal category
        \item $\text{Mor}_\mathcal{C}(A,B)$ (for any 2 objects A,B) is a chain complex.
   \end{itemize} 
    \end{Def}

\begin{Def}
Let $(\mathcal{C}, \square, I)$ and $(\mathcal{D}, \odot, J)$ be monoidal categories. A \textit{monoidal functor}  $ 
\mathscr{F} : (\mathcal{C}, \square, I) \to (\mathcal{D}, \odot , J)$ is a functor $ 
\mathscr{F} : \mathcal{C} \to \mathcal{D}$ equipped with the following structure,
\begin{enumerate}
    \item a morphism $\epsilon : J \to F(I)$
    \item a natural transformation $\mu$ from the functor $ \odot \cdot (F,F) : \mathcal{C} \cross \mathcal{C} \to \mathcal{D} \cross \mathcal{D} \to \mathcal{D}$ to the functor $ F \cdot \square : \mathcal{C} \cross \mathcal{C} \to \mathcal{C} \to \mathcal{D}$ with components of the form (for objects or morphisms A,B in $\mathcal{C}$),
    \begin{equation*}
        \mu_{A,B} : F(A) \odot F(B) \to F(A \square B)
    \end{equation*}
    \item Two conditions on the above morphisms
    \begin{itemize}
        \item For all objects A,B,C in $\mathcal{C}$, the associators $\alpha, \beta$, the following diagram commutes (\textit{associativity}),

\begin{center} \begin{tikzcd}[sep=huge]
  (F(A) \odot F(B)) \odot F(C) \arrow[r,"\beta_{F(A),F(B),F(C)}"] \arrow[d,"\mu_{A,B} \odot 1_{F(C)}"'] & F(A) \odot (F(B) \odot F(C)) 
 \arrow[d,"1_{F(A)} \odot \mu_{B,C} "] \\
  F(A \square B) \odot F(C)  \arrow[d,"\mu_{A \square B, C} "'] & F(A) \odot F(B \square C) \arrow[d,"\mu_{A,B \square C}"] \\
  F((A \square B )\square C) \arrow[r,"F(\alpha_{A,B,C})"'] &
   F(A \square (B \square C)) 
\end{tikzcd}\\
\end{center}
        \item For all objects/morphisms A in $\mathcal{C}$, the unitors $\lambda,\rho,\Lambda,\varrho$ the following diagrams commute (\textit{unitality}),

\begin{center} \begin{tikzcd}[sep=huge]
  J \odot F(A) \arrow[r,"\epsilon \odot 1_{F(A)}"] \arrow[d,"\Lambda_{F(A)}"'] & F(I) \odot F(A) \arrow[d,"\mu_{I,A}"] &  F(A) \odot J \arrow[r,"1_{F(A)} \odot \epsilon"] \arrow[d,"\varrho_{F(A)}"'] & F(A) \odot F(I) \arrow[d,"\mu_{A,I}"]  \\
  F(A)  & F(I \square A) \arrow[l,"F(\lambda_{A})"] & F(A)  & F(A \square I) \arrow[l,"F(\rho_{A})"]
\end{tikzcd}\\
 \end{center}       
    \end{itemize}
\end{enumerate}
\end{Def}

\begin{Def}
    Let $(\mathcal{C}, \square, I,R)$ and $(\mathcal{D}, \odot, J,S)$ be braided monoidal categories ($R,S$ are braidings). A \textit{braided monoidal functor}  $ 
\mathscr{F} : (\mathcal{C}, \square, I,R) \to (\mathcal{D}, \odot , J,S)$ is a monoidal functor  $ 
\mathscr{F}, \mu, \epsilon : (\mathcal{C}, \square, I) \to (\mathcal{D}, \odot , J)$ which respects the braiding on both sides, i.e., satisfies the law,

\begin{center} \begin{tikzcd}[sep=huge]
  F(A) \odot F(B) \arrow[r,"S_{F(A),F(B)}"] \arrow[d,"\mu_{A,B}"'] & F(B) \odot F(A) \arrow[d,"\mu_{B,A}"]  \\
  F(A \square B ) \arrow[r,"F(R_{A,B})"]  & F(B \square A) 
\end{tikzcd}\\
\end{center}
If $\mathcal{C},\mathcal{D}$ happen to be symmetric monoidal categories, then the braided monoidal functor $F$ is said to be a \textit{symmetric monoidal functor}.
\end{Def}

\begin{Def}
\label{24}
 Let $\mathcal{C}$ and $\mathcal{D}$ be differential graded symmetric monoidal categories. A \textit{differential graded symmetric monoidal} (dgsm) functor from $ \mathcal{C}$ to $\mathcal{D}$ is a functor $F : \mathcal{C} \to \mathcal{D}$ which is,
 \begin{itemize}
     \item a symmetric monoidal functor 
     \item a chain map $F : \text{Mor}_\mathcal{C}(A,B) \to \text{Mor}_\mathcal{D}(F(A),F(B))$ (for any 2 objects A,B in $\mathcal{C}$)
 \end{itemize}
\end{Def}

\begin{proof} of Proposition \ref{37}
\begin{itemize}
    \item The bifunctor is given by tensor product ($\otimes$) of chain complexes over the ground field $\R$, $I \equiv \R$
    We assume that disjoint union of chains is associative (identifying upto isomorphism). Hence,\\
$\alpha, \lambda, \rho$ are identity maps on components,
\begin{equation*}
    \alpha : (a \otimes b) \otimes c \mapsto a \otimes (b \otimes c) = (a \otimes b) \otimes c = a \otimes b \otimes c
\end{equation*}
\item Pentagon and triangle identities are trivially satisfied as this is a strict category.
    \item Define $S_{A,B} : A \otimes B \to B \otimes A$ as $a \otimes b \mapsto b \otimes a$ which clearly squares to identity, and satisfies hexagon identity,

  \begin{center} \begin{tikzcd}[sep=huge]
  (A \otimes B) \otimes C \arrow[r,"\alpha_{A,B,C}"] \arrow[d,"S_{A,B} \otimes 1_C"'] & A \otimes (B \otimes C)\arrow[r,"S_{A,B \otimes C}" ] & (B \otimes C) \otimes A \arrow[d,"\alpha_{B,C,A}"]\\
  (B \otimes A) \otimes C \arrow[r,"\alpha_{B,A,C}"'] &
  B \otimes (A \otimes C) \arrow[r,"1_B \otimes S_{A,C}"']  &
  B \otimes (C \otimes A) 
  \end{tikzcd}\\
\end{center}

We wish to check,
\begin{equation*}
\alpha_{B,C,A} \cdot S_{A,B \otimes C} \cdot \alpha_{A,B,C} [(a \otimes b) \otimes c ]
    =
    1_B \otimes S_{A,C} \cdot \alpha_{B,A,C} \cdot S_{A,B} \otimes 1_C
    [(a \otimes b) \otimes c]
\end{equation*}

The L.H.S. gives,
\begin{equation*}
    \alpha_{B,C,A} \cdot S_{A,B \otimes C} [a \otimes (b \otimes c) ] = \alpha_{B,C,A} [(b \otimes c) \otimes a] = [b \otimes (c \otimes a)]   
\end{equation*}
The R.H.S. gives,
\begin{equation*}
    1_B \otimes S_{A,C} \cdot \alpha_{B,A,C} 
    [(b \otimes a) \otimes c] =  1_B \otimes S_{A,C} 
    [b \otimes (a \otimes c)] = 
    [b \otimes (c \otimes a)] = L.H.S.
\end{equation*}
\item $\text{Mor}_\text{Comp}(A,B) := \text{Hom}_\bullet(A,B)$ forms a chain complex in the following sense,\\
Let $\text{Hom}_p(A,B)$ consist of all degree p linear maps from $A \to B$. Now, given $\Phi \in \text{Hom}_p(A,B)$, define,
\begin{equation*}
    d_{\text{Comp}} \Phi  : = d_B \cdot \Phi - (-1)^{|\Phi|} \Phi \cdot d_A \in \text{Hom}_{p-1}(A,B)
\end{equation*}
\end{itemize}
\end{proof}

\bibliographystyle{amsalpha}    
\bibliography{fullpaper}

\end{document}